\documentclass[letterpaper, 10 pt, conference]{ieeeconf}  

  \IEEEoverridecommandlockouts                              

\usepackage{etoolbox}
\usepackage{amsmath}
\usepackage{amssymb}
\usepackage{comment}
\usepackage{caption}
\usepackage{subcaption}
\usepackage{color}
\usepackage{textcomp}
\usepackage{paralist}
\usepackage{graphicx}      
\usepackage{cite}
\usepackage{url}
\usepackage{array}
\usepackage{makecell}
\usepackage[export]{adjustbox}
\usepackage{ntheorem}

\usepackage{mathtools}
\usepackage[utf8x]{inputenc} 
\usepackage[T1]{fontenc} 
\usepackage{bm} 

\usepackage{tikz}
\usetikzlibrary{matrix} 
\usetikzlibrary{calc} 
\DeclareMathOperator{\sgn}{sgn}

\tikzstyle{block} = [draw,rectangle,thick,minimum height=2em,minimum width=2em]
\tikzstyle{sum} = [draw,circle,inner sep=0mm,minimum size=2mm]
\tikzstyle{line} = [thick]
\tikzstyle{branch} = [circle,inner sep=0pt,minimum size=1mm,fill=black,draw=black]
\tikzstyle{guide} = []
\usepackage{fontawesome} 
\usetikzlibrary{positioning}
 
\renewcommand{\vec}[1]{\ensuremath{\boldsymbol{#1}}} 

\usetikzlibrary{graphs,graphs.standard,quotes}
\usetikzlibrary{arrows,decorations.pathmorphing,backgrounds,positioning,fit,petri}
\usetikzlibrary{shapes,arrows}
\usepackage{verbatim}

\newtheorem{theorem}{\textbf{Theorem}}

\newtheorem{myprob}{\textbf{Problem}}

\newtheorem{mylemma}{\textbf{Lemma}}

\newcommand\norm[1]{\left\lVert#1\right\rVert}
\title{\LARGE \bf
Robust DAC for a Network of Agents with Time-varying Reference Signals
}
\author{Solomon Gudeta, Ali Karimoddini, and  Mohammadreza Davoodi  
\thanks{Solomon Gudeta and Ali Karimoddini are with the Department of Electrical and Computer Engineering,
North Carolina Agricultural and Technical State University, Greensboro, North Carolina, USA.}
\thanks{Mohammadreza Davoodi is with the UT Arlington Research Institute, University of Texas at Arlington, Fort Worth, Texas, USA.}
\thanks{Corresponding author: A. Karimoddini, Tl: +13362853313, {\tt\small akarimod@ncat.edu}.}
}

\begin{document}
\maketitle
\thispagestyle{empty}
\pagestyle{empty}

\begin{abstract}
This paper presents continuous dynamic average consensus (DAC) algorithms for a group of agents to estimate the average of their time-varying reference signals cooperatively. We propose consensus algorithms that are robust to agents joining and leaving the network, at the same time, avoid the chattering phenomena and guarantee zero steady-state consensus error. Our algorithms are edge-based protocols with smooth functions in their internal structure to avoid the chattering effect. Furthermore, each agent is only capable of performing local computations and can only communicate with its local neighbors. For a balanced and strongly connected underlying communication graph, we provide the convergence analysis to determine the consensus design parameters that guarantee the agents' estimate of their average to asymptotically converge to the average of the time-varying reference signals of the agents.  We provide simulation results to validate the proposed consensus algorithms and to perform a performance comparison of the proposed algorithms to existing algorithms in the literature. 
\end{abstract}

\section{Introduction}
Consensus protocols are becoming the backbone of various distributed networked systems \cite{lynch1996distributed,olfati2007consensus}. Consensus protocols are the rules through which the agents in a networked system interact to reach to an agreement on quantities of interest through local communications with their neighbors. Quantities of interest might take different forms, such as the average, the minimum, the maximum, and the min-max of local information at every agent in the network, based on which many researchers design various average consensus, min-consensus, max-consensus, and min-max consensus algorithms, respectively \cite{saber2003consensus}. These consensus algorithms play vital roles in a myriad of applications, such as formation control \cite{freeman2006distributed,yang2008multi,mallik2016scalable,ren2008distributed}, sensor-fusion \cite{spanos2005distributed}, distributed tracking \cite{song2010tracking}, distributed optimization \cite{tsianos2012consensus}, resource allocation \cite{li2017distributed}, network connectivity maintenance \cite{dimarogonas2010bounded}, and many others. 

In this paper, we focus on the average consensus problem--the problem of reaching an agreement on the average of local reference signals at each agent in a distributed way. The reference signals at each agent could be sensor measurements, local computations, states of the agents, and states of the leader, to name a few.  Commonly, an average consensus problem has two forms: the static average consensus and the DAC. In the static average consensus, the agents use the reference signals only to initialize the consensus iteration algorithm. Then, the agents use linear weighting to update the estimate of the average iteratively \cite{olfati2004consensus}. In contrast, in the dynamic consensus, the signals at each agent derive the consensus iteration algorithm continuously \cite{spanos2005dynamic}. Focusing on the DAC, many authors have studied dynamic average tracking (DAT) problems \cite{rahili2017heterogeneous, zhao2017distributed, ghapani2016distributed,chen2012distributed} and DAC  problems \cite{spanos2005dynamic, moradian2017dynamic, kia2014dynamic, zhu2010discrete, george2019robust, george2017robust} extensively. DAT problems involve the design of combined estimator and controller algorithms to track a time-varying average of agents' reference signals. On the other hand, in DAC problems, the estimator and control design problems are considered separately.

In the DAC, agents in the network estimate the average of local time-varying signals at each agent in the network via local communications with its neighboring agents. The authors in \cite{spanos2005dynamic} attempted this problem by proposing a distributed robust consensus algorithm for reaching an average consensus on the average of signals with constant values in the presence of non-uniform delays in the network. In \cite{freeman2006stability}, proportional and proportional-integral estimators are proposed to estimate the average of slowly varying signals. However, the performance of the estimators in \cite{spanos2005dynamic} and \cite{freeman2006stability} deteriorate when the reference signals are fast time-varying signals. In \cite{bai2010robust}, the authors extended the proportional integrator estimator to estimate the average of multiple time-varying signals, including some classes of polynomial and sinusoidal signals, at each agent in the network. Leveraging the singular perturbation theory, the authors in \cite{kia2013singularly} proposed two DAC algorithms with exponential convergence for any initial conditions. However, the implementation of consensus algorithms in \cite{bai2010robust} and \cite{kia2013singularly} requires the knowledge of the model of reference signals and the first and second derivatives of the reference signals, respectively.   

Recently, employing a switching control, powerful non-linear consensus protocols that converge in finite time are proposed in \cite{george2017robust, george2019robust, chen2012distributed}. However, the consensus algorithm in \cite{chen2012distributed} is not robust to topology changes introduced by agents leaving and joining the network. In contrast, the DAC algorithms in \cite{george2017robust, george2019robust} are robust to topology changes and do not assume access to derivatives of reference signals except to determine or reset consensus gains. However, due to the discontinuity introduced by the switching control, the consensus algorithms in \cite{george2017robust, george2019robust} suffer from the chattering effect. The chattering phenomenon is undesired behavior, and reducing its effects require very small integration time-steps. Therefore, a discrete-time implementation of algorithms with switching control requires high sampling rates to decrease the adverse effects of the chattering phenomena. In fact, authors in \cite{george2017robust, george2019robust} were aware of the adverse effect of the chattering phenomena and proposed a boundary-layer approximation (approximating the discontinuous switching signal by a smooth function inside a boundary layer) to remove the chattering effect. However, the boundary layer approximation guarantees only the convergence of the consensus error to an $\epsilon-$neighbourhood of the origin in finite time. In \cite{stamouli2019robust}, a new continuous dynamic consensus algorithm that avoids the chattering phenomena and, at the same time,  does not require the knowledge of the derivatives of reference signals is proposed.  However, the convergence of the consensus algorithm in \cite{stamouli2019robust} is only limited to guaranteeing a bounded steady-state error. 

In this paper, we introduce two continuous robust DAC algorithms for a network of dynamic agents with time-varying reference signals where the underlying network typologies are balanced and strongly connected bidirectional graph. The algorithms are edge-based algorithms where the edges capture the disagreement between agent $i$ and its neighboring agents. Then, we design the internal structure of the consensus estimator on each agent based on the edges via a smooth function. In the first algorithm, we propose a robust protocol that removes the chattering effect, at the same time, guarantees an adjustable bounded steady-state error. The proposed algorithm does not require knowing the time derivatives of the reference signals. In the second algorithm, unlike in many DAC algorithms (for example, see \cite{chen2012distributed,kia2013singularly}), we utilize the knowledge of the derivatives of reference signals only when it is  available. In other words, we can convert the proposed algorithm to an equivalent protocol that does not require the knowledge of the derivatives reference signals through a coordinate transformation when the knowledge of the derivatives of reference signals is not available. Also, the consensus protocol proposed in the second algorithm i) is robust to topology changes, ii) does not suffer from the chattering phenomena, and iii) guarantees zero steady-state error. 

Compared to recent methods in \cite{george2017robust, george2019robust}, our algorithms do not involve discontinuous switching signal. In practice, it is impossible to implement the discontinuous switching signal precisely and therefore,  methods in \cite{george2017robust, george2019robust} guarantees only the convergence of the consensus error to an $\epsilon-$neighbourhood of the origin. Moreover, our second algorithm guarantees asymptotic convergence of the consensus error to the origin versus the algorithm in \cite{stamouli2019robust} that only guarantees a bounded steady-state error and the methods in \cite{george2017robust, george2019robust} that guarantee a small steady-state error in reality. Simulation results are provided to compare the performance of the proposed consensus algorithm in this paper with  \cite{george2017robust}.

\section{Mathematical preliminaries and Problem formulation}
\label{sec:form_and_tra_prob}

For a  bidirectional graph $\mathcal{G}(t)$ $=$ $(\mathcal{V}(t),\mathcal{E}(t))$, where $\mathcal{V}(t) \triangleq \{v_1(t),\cdots,v_N(t)\}$ is the set of agents and $\mathcal{E}(t) = \{\mathcal{E}_1(t),\cdots,\mathcal{E}_k\} \subseteq \mathcal{V}(t) \times \mathcal{V}(t)$--$\{(v_i(t),v_i(t)|v_i(t)\in\mathcal{V}(t)\}$ is the set of communication links among the agents. The set of neighbors $\mathcal{N}_i(t)$ of agent $i$, $i = \{1,\cdots,N\}$ is: $
\mathcal{N}_i(t) = \{j \in \mathcal{V}(t):(j,i)\in\mathcal{E}(t)\}$.
Let $\mathcal{A}(t)= [a_{ij}]\in \{0,1\}^{N\times N}$ be the adjacency matrix of graph $\mathcal{G}(t)$, where $a_{ij} = 1$ if $(v_i(t),v_j(t)) \in \mathcal{E}(t)$ and $a_{ij} = 0$ otherwise. Then, the degree matrix and the laplacian matrix of graph $\mathcal{G}(t)$ is given by $\Delta(t) = \textrm{diag}(\mathcal{A}\mathbf{1}_N)$ and  $\mathcal{L}(t) = \Delta(t)-A(t)$, respectively. The notation $\mathbf{1}_N$ denotes $N-$dimensional vector of all ones. The incidence matrix of graph $\mathcal{G}(t)$ is given by $\mathcal{B}(t) = [{b_{ij}}], \in \{-1,0,1\}^{N \times k}$, where $b_{ij} = -1$ for the outgoing communication link from agent $i$, $b_{ij} = 1$ for the incoming communication link to agent $i$, and $b_{ij} = 0$ otherwise. The following two lemmas will be used in the ensuing sections.
\begin{mylemma}\cite{george2019robust} Let $M \triangleq (I_N - \frac{1_N1_N^\mathrm{T}}{N})$. For any balanced and strongly connected bidirectional graph $\mathcal{G}(t)$, the Laplacian matrix $\mathcal{L}(t)$ and the incidence matrix $\mathcal{B}(t)$ satisfy
$
M = \mathcal{L}(t)(\mathcal{L}(t))^+ = (\mathcal{B}(t)\mathcal{B}^\mathrm{T}(t))(\mathcal{B}(t)\mathcal{B}^\mathrm{T}(t))^+$, where $(.)^+$ is the generalized inverse. 
\label{lemma:lapinci}
\end{mylemma}
From Lemma \ref{lemma:lapinci} it is also clear that the Laplacian matrix $\mathcal{L}(t)$ and the incidence matrix $\mathcal{B}(t)$ of graph $\mathcal{G}(t)$ are related as $\mathcal{L}(t) = \mathcal{B}(t) \mathcal{B}^\mathrm{T}(t)$. 
\begin{mylemma} 
Laplacian matrix of a balanced and strongly connected bidirectional graph $\mathcal{G}(t)$ is a positive definite matrix with an eigenvalue at $0$ corresponding to right and left eigenvectors $\mathbf{1}_N$ and $\mathbf{1}_N^\mathrm{T}$, respectively.
\label{lemma:lap_eig}
\end{mylemma}
\begin{proof}
See \cite{marsden2013eigenvalues}.
\end{proof}
Lemma \ref{lemma:lap_eig} describes that $\forall x \in R^N$, $x^\mathrm{T}\mathcal{L}(t)x \geq 0$, $\mathbf{1}_N^T\mathcal{L}(t) = 0$, $\mathbf{1}_N^T\mathcal{B}(t) = 0$, $\mathbf{1}_N^TM = 0$, $\mathcal{L}(t)\mathbf{1}_N = 0$, $\mathcal{B}(t)\mathbf{1}_N = 0$ and $M\mathbf{1}_N = 0$.

Now, consider a network of $N$ agents, where each agent $i$, $i = \{1,\cdots,N\}$, computes or measures a time-varying reference signal $z_i(t) \in \mathbb{R}$ with bounded first derivative $\dot{z}_i(t)$, i.e., $\sup_{t\geq t_0}{\norm{\dot{z}_i(t)}_\infty} \leq \mathcal{\psi}_i$, where $\psi_i$ is a positive constant. Let the underlying network topology at time $t$ be given by a balanced and strongly connected bidirectional graph $\mathcal{G}(t)$, and
let the average of the agents' time-varying reference signals be: $
\bar{z}(t) = \frac{1}{N}\sum_{i=1}^Nz_i(t) = \frac{1}{N}\mathbf{1}_N\mathbf{1}_N^\mathrm{T}z(t)
$,
where $z(t) = [z_1^\mathrm{T} \cdots z_N^\mathrm{T}]^\mathrm{T}$. Our main objective is to design consensus algorithms that the agents use to cooperatively estimate a time-varying average signal $\bar{z}(t)$ in a distributed fashion through local communications with their respective neighbors in the network. More precisely, we state the main problem as follows.
\begin{myprob}
Consider a network of $N$ agents where the underlying network topology at time $t$ is given by a balanced and strongly connected bidirectional graph. Let $\gamma_i(t)$ be agent $i$'s estimate of the average signal $\bar{z}(t)$, where $i = \{1,\cdots,N\}$. Let agent $i$'s estimation error and agent$i$'s disagreement with agent $j$ are computed as $\tilde{\gamma}_i(t) = \gamma_i(t)-\bar{z}(t)$, and  $ \tilde{\gamma}_i(t)-\tilde{\gamma}_j(t) = \gamma_i(t)-\gamma_j(t)$, where $ j \in \mathcal{N}_i(t)$, respectively. Then, design a consensus protocol such that $\displaystyle {\lim_{t \to \infty}\tilde{\gamma}_i(t) \rightarrow 0}$, that is, i) $\sum_{i = 1}^N \tilde{\gamma}_i(t) = 0$, and ii) $\displaystyle {\lim_{t \to \infty}(\tilde{\gamma}_i(t) - \tilde{\gamma}_j(t)) \rightarrow 0}$, $\forall (i, j)\in \mathcal{E}(t)$, $t \geq t_0$.
\label{prob:cons}
\end{myprob}
\section{Proposed robust DAC}
\label{sec:form_and_tra_framework}
In this section, we present two DAC algorithms to solve problem \ref{prob:cons}. In both algorithms, similar to \cite{george2019robust}, we employ edge-based approaches to design the consensus protocols. Unlike the node-based approaches, in the edge-based approaches, there are multiple internal states per each node in the network. 
\subsection{Robust DAC algorithm I}
Let agent $i$'s estimation error and agent$i$'s disagreement with agent $j$ are computed as $\tilde{\gamma}_i(t) = \gamma_i(t)-\bar{z}(t)$, and  $ \tilde{\gamma}_i(t)-\tilde{\gamma}_j(t) = \gamma_i(t)-\gamma_j(t)$, where $ j \in \mathcal{N}_i(t)$, respectively. Now, suppose that each agent in the network implements the following consensus protocol:
\begin{equation}
\begin{split}
\dot{\eta}_{ij}^+(t) &= -\rho\tanh\{c(\gamma_{i}(t)-\gamma_{j}(t))\}\\
\dot{\eta}_{ij}^-(t) &=
-\rho\tanh\{c(\gamma_{j}(t)-\gamma_{i}(t))\}\\
\gamma_{i}(t) &= \sum_{j\in \mathcal{N}_i}{\eta}_{ij}^+(t)-\sum_{j\in \mathcal{N}_i}{\eta}_{ij}^-(t)+z_{i}(t)\\ \eta_{ij}^+(t_0) &= \eta_{ij_{0}}, \eta_{ij}^-(t_0) = \eta_{ij_{0}}, c\geq 1, j\in \mathcal{N}_i,
\end{split}
\label{eq:consensus_smci}
\end{equation}
\noindent where ${\eta}_{i} = [{\eta}_{ij}^+(t) \quad {\eta}_{ij}^-(t)]^\mathrm{T} \in \mathbb{R}^{2\mathcal{N}_i}$ is the internal state of the estimator of agent $i$, $\rho \in \mathbb{R}$ and $c \in \mathbb{R}$ are global consensus parameters, and  $\gamma_{i}(t) \in \mathbb{R}$ is agent $i$'s estimate of $\bar{z}(t)$. The edge dynamics is captured via the internal state dynamics of the estimator. From (\ref{eq:consensus_smci}), it is clear that the edge dynamics captures the state of the disagreement between agent $i$ and agent $j$. This approach makes the protocol robust to agents joining or leaving the network, and communication link failures among the agents. 

To simplify the proof of convergence of the proposed consensus algorithm, we collect together (\ref{eq:consensus_smci}) into the following compact form:
\begin{equation}
\dot{\eta}(t) = -\rho\tanh\{c \mathcal{B}^\mathrm{T}(t)\gamma(t)\},
\gamma(t) = \mathcal{B}(t)\eta(t)+z(t),
\label{eq:consensus_smc}
\end{equation}
where $\eta(t)=[{\eta_1}, \cdots, {\eta_N}]^\mathrm{T}$, $\eta(t_0) = \eta_{0}$, $\gamma(t)=[{\gamma_1}, \cdots, {\gamma_N}]^\mathrm{T}$, $z(t)=[{z_1}, \cdots, {z_N}]^\mathrm{T}$, $\mathcal{B}(t)$ is the incidence matrix of graph $\mathcal{G}(t)$, and the $tanh(.)$ is defined component wise. The $tanh(.)$ function in (\ref{eq:consensus_smc}) empowers the consensus protocol to avoid the chattering phenomena. 

\begin{theorem}
Let $V$ be a smooth positive definite function and suppose that the sets 
\begin{equation}
\Omega_{d_c} = \{V \leq d_c\}, \Omega_\epsilon = \{V \leq \epsilon\}, \Lambda = \{\epsilon \leq V \leq d_c\}
\end{equation} are invariant set for some $d_c \geq \epsilon \geq \frac{\delta^2}{2}$. For the balanced and strongly connected bidirectional graph $\mathcal{G}(t)$ and  time-varying reference signals $z_i(t)$, $i =\{1,\cdots,N\}$ with bounded first derivatives, the robust DAC algorithm in (\ref{eq:consensus_smc})
guarantees that $\tilde{\gamma}(t) = |\gamma(t) - \frac{1_N1_N^\mathrm{T}}{N}z(t)|$ is uniformly ultimately bounded and converges to the adjustable compact set 
\begin{equation}
\{\tilde{\gamma}(t):
 \norm{\tilde{\gamma}(t)}_\infty \leq \norm{\left(\mathcal{B}^\mathrm{T}(t)\right)^+}_2\left( \sqrt{\frac{N\lambda_{max}\left(\mathcal{L}(t)\right)}{\lambda_2\left(\mathcal{L}(t)\right)}}\right)\delta
\},
\end{equation} where $\delta = \frac{1}{2c} ln\left(\frac{\rho+\norm{\dot{z}(t)}_1}{\rho-\norm{\dot{z}(t)}_1}\right)$ in finite time $
t^* = t_0+\frac{d_c-\epsilon}{b}$ ,
for all $\tilde{\gamma}(t_0)$, if and only if the global consensus parameters are selected such that $\rho > \sup_{t\geq t_0}{\norm{\dot{z}(t)}_1}$ and $c \geq 1$.
\label{theorem:graph1}
\end{theorem}

\begin{proof}  Using the dynamic consensus estimator  in (\ref{eq:consensus_smc}), the estimator error dynamics can be found  as
\begin{equation}
\dot{\tilde{\gamma}}(t) = \mathcal{B}(t)\dot{\eta}(t)+M\dot{z}(t), \quad M = (I_N - \frac{1_N1_N^\mathrm{T}}{N})
\label{eq:consensus_er}
\end{equation}
Consider the Lyapunov function $V \triangleq \frac{1}{2}\tilde{\gamma}^\mathrm{T}(t)\tilde{\gamma}(t)$. Taking a derivative of $V$, we have $\dot{V} = -\vartheta^\mathrm{T}(t)\rho\tanh(c\vartheta(t))+\vartheta^\mathrm{T}(t)\mathcal{B}^\mathrm{T}(t)\left(\mathcal{B}(t)\mathcal{B}^\mathrm{T}(t)\right)^+\dot{z}(t)$, where $\vartheta(t) = \mathcal{B}^\mathrm{T}(t)\tilde{\gamma}(t)=\mathcal{B}^\mathrm{T}(t)\gamma(t)$ . Expanding $\dot{V}$, we will have 
\begin{multline}
\dot{V} = -\rho[{\vartheta_1}(t)\tanh(c{\vartheta_1}(t))+\cdots+{\vartheta_l}(t)\tanh(c{\vartheta_l}(t))]\\+{\vartheta_1}(t)(m_{11}\dot{z}_{1}(t)+\cdots+m_{1N}\dot{z}_{N}(t))+\cdots\\+{\vartheta_l}(t)(m_{l1}\dot{z}_{1}(t)+\cdots+m_{lN}\dot{z}_{N}(t))
\label{eq:vdot}
\end{multline}
where $m_{ij} \in \mathcal{B}^\mathrm{T}(t)\left(\mathcal{B}(t)\mathcal{B}^\mathrm{T}(t)\right)^+$, $i={1,\cdots,l}$, $j = {1,\cdots,N}$, $l$ is the total number of edges in $\mathcal{G}(t)$. Let $\Theta_i(t) = (m_{i1}\dot{\bar{z}}_{1}(t)+\cdots+m_{iN}\dot{\bar{z}}_{N}(t))$. Then, (\ref{eq:vdot}) can be re-written as
\begin{multline}
\dot{V} = -\rho[{\vartheta_1}(t)\tanh(c{\vartheta_1}(t))+\cdots+{\vartheta_l}(t)\tanh(c{\vartheta_l}(t))]\\+{\vartheta_1}(t)\Theta_1(t)+\cdots+{\vartheta_l}(t)\Theta_l(t)
\label{eq:vdot1}
\end{multline}
Letting $w_i(\vartheta(t)) = -\rho{\vartheta_i}(t)\tanh(c{\vartheta_i}(t))+{\vartheta_i}(t)\Theta_i(t)$, we have $\dot{V} = \sum_{i=1}^lw_i(\vartheta(t))$.  Rewriting $\tanh(.)$ in terms of $\sgn(.)$ results in
\begin{multline}
w_i(\vartheta(t)) = -\rho{\vartheta_i}(t)\left(\sgn(c{\vartheta_i}(t))\left(1-\frac{2}{e^{2c|{\vartheta_i}(t)|}+1}\right)\right)\\+{\vartheta_i}(t)\Theta_i(t)
\end{multline}
Then, since ${\vartheta_i}(t)\sgn(c{\vartheta_i}(t)) = |{\vartheta_i}(t)|$, we have
\begin{equation}
w_i(\vartheta(t)) \leq -|{\vartheta_i}(t)|(\rho-\frac{2\rho}{e^{2c|{\vartheta_i}(t)|}+1}-|\Theta_i(t)|)
\label{eq:wbound}
\end{equation}
Since $|\Theta_i(t)| \leq \sum_{i=1}^N|\dot{z}_{i}(t)|$ $=\norm{\dot{z}(t)}_1$
, from (\ref{eq:vdot1}) and (\ref{eq:wbound}), it immediately follows that $\dot{V} = \sum_{i=1}^lw_i(\vartheta(t)) < 0$ for all $\norm{{\vartheta}(t)}_\infty > \delta $, $\rho >\norm{\dot{z}(t)}_1$,  and $c\geq 1$, where
\begin{equation}
 \delta = \frac{1}{2c} ln\left(\frac{\rho+\norm{\dot{z}(t)}_1}{\rho-\norm{\dot{z}(t)}_1}\right).
 \label{eq:bounddelta}
 \end{equation} 
Choosing $\epsilon$ and $d_c$ such that $\frac{\delta^2}{2} < \epsilon < d_c$, then $\dot{V}$ is negative in the invariant set $\Lambda = \{\epsilon \leq V \leq d_c\}$.
 Employing Courant-Fischer Theorem, we have  $\frac{\lambda_2}{N}\left(\mathcal{L}(t)\right)\norm{\tilde{\gamma}(t)}_2^2$ $\leq$
 $V$ $\leq$ $\lambda_{max}\left(\mathcal{L}(t)\right)\norm{\tilde{\gamma}(t)}_2^2$, where $\lambda_2(\mathcal{L}(t))$ is the algebraic connectivity of $\mathcal{G}(t)$ and $\lambda_{max}(\mathcal{L}(t))$ is the maximum eigenvalues of $\mathcal{L}(t)$. Since $V$ is radially unbounded, to determine the ultimate bound for $\norm{\vartheta(t)}_\infty$,  we define $\alpha_1(r),\alpha_2(r)\in \mathcal{K}_\infty$ as follows
 \begin{equation}
 \alpha_1(r) = \frac{\lambda_2\left(\mathcal{L}(t)\right)}{N}r^2, \quad 
 \alpha_2(r) = \lambda_{max}\left(\mathcal{L}(t)\right)r^2
 \end{equation}
 Then, the ultimate bound for $\norm{\vartheta(t)}_\infty$ is given as
 \begin{equation}
 \norm{\vartheta(t)}_\infty \leq \alpha_1^{-1}(\alpha_2(\delta)) =\left( \sqrt{\frac{N\lambda_{max}\left(\mathcal{L}(t)\right)}{\lambda_2\left(\mathcal{L}(t)\right)}}\right)\delta
 \end{equation}
 Based on this, the ultimate bound on the consensus error $\tilde{\gamma}(t)$ is given as
 \begin{equation}
 \norm{\tilde{\gamma}(t)}_\infty \leq \norm{\left(\mathcal{B}^\mathrm{T}(t)\right)^+}_2\left( \sqrt{\frac{N\lambda_{max}\left(\mathcal{L}(t)\right)}{\lambda_2\left(\mathcal{L}(t)\right)}}\right)\delta
 \label{eq:conserror}
 \end{equation}
Choosing $\rho \gg \sup_{t\geq t_0}{\norm{\dot{z}(t)}_1}$, then from (\ref{eq:bounddelta}), we have $\delta\approx 0$. Then from (\ref{eq:conserror}), we can conclude that the consensus error can be adjusted to $ \norm{\tilde{\gamma_k}(t)}_\infty \approx 0$ by increasing the global consensus parameter $\rho$. However, one has to exercise precaution when tuning $\rho$ to a very large values since a very large $\rho$ could amplify noise and consensus disagreement.

Let $ 
b = \min_{\vartheta \in \Lambda} \{-\sum_{i=1}^lw_i(\vartheta(t))\} > 0$, 
$\textrm{s.t.}$  $\rho >\norm{\dot{z}(t)}_1$.
Then, $\dot{V} \leq -b$, $\forall \vartheta(t) \in \Lambda$, $\forall t \geq t_0 \geq 0$. Accordingly, $V \leq V_0 - b(t-t_0) \leq d_c-b(t-t_0)$, and the consensus error enters the adjustable compact set within finite time $ t^* = t_0+\frac{d_c-\epsilon}{b}$. This concludes the proof.
\end{proof}

The consensus in a network of agents implementing the algorithm in (\ref{eq:consensus_smc}) converges to the average of the reference signals of the agents in the network with a bounded and adjustable steady-state error. Also, from the consensus error dynamics in (\ref{eq:consensus_er}), we have$
\mathbf{1}_N^\mathrm{T}\tilde{\gamma}(t)  = 0$, 
$\mathbf{1}_N^\mathrm{T}\dot{\tilde{\gamma}}(t)  = 0$,
as $\mathbf{1}_N^\mathrm{T}\mathcal{B}(t)\eta(t)=0$ and $\mathbf{1}_N^\mathrm{T}Mz(t)=0$ for all $t \geq t_0$. Therefore, the consensus protocol in (\ref{eq:consensus_smc}) does not require special initialization requirement.
\subsection{Robust DAC algorithm II}
In this section, we propose an enhanced DAC algorithm that leverages the knowledge of the derivative of reference signals, if available. When the derivative of reference signals is not available, we convert the proposed algorithm into another DAC protocol utilizing coordinate transformation. In this case, the transformed DAC algorithm does not require the knowledge of the derivative of reference signals. Let agent $i$'s estimation error and agent$i$'s disagreement with agent $j$ are computed as $\tilde{\gamma}_i(t) = \gamma_i(t)-\bar{z}(t)$, and  $ \tilde{\gamma}_i(t)-\tilde{\gamma}_j(t) = \gamma_i(t)-\gamma_j(t)$, where $ j \in \mathcal{N}_i(t)$, respectively. Now, consider that each agent implements a DAC algorithm  of the form

\small
\begin{equation}
\begin{split}
\dot{\eta}_{ij}^+(t) &= -\alpha(\dot{z}_i(t)-\dot{z}_j(t))-\rho\tanh\{c(\gamma_{i}(t)-\gamma_{j}(t))\}\\
\dot{\eta}_{ij}^-(t) &= -\alpha(\dot{z}_j(t)-\dot{z}_i(t)
-\rho\tanh\{c(\gamma_{j}(t)-\gamma_{i}(t))\}\\
\gamma_{i}(t) &= \sum_{j\in \mathcal{N}_i}{\eta}_{ij}^+(t)-\sum_{j\in \mathcal{N}_i}{\eta}_{ij}^-(t)+z_{i}(t)\\ \eta_{ij}^+(t_0) &= \eta_{ij_{0}}, \eta_{ij}^-(t_0) = \eta_{ij_{0}}, c\geq 1, j\in \mathcal{N}_i,
\end{split}
\label{eq:consensus_smci1}
\end{equation}
\normalsize
\noindent where ${\eta}_{i} = [{\eta}_{ij}^+(t) \quad {\eta}_{ij}^-(t)]^\mathrm{T} \in \mathbb{R}^{2\mathcal{N}_i}$ is the internal estimator state; $\alpha$, $\rho \in \mathbb{R}$ and $c \in \mathbb{R}$ are the global design parameters, and  $\gamma_{i}(t) \in \mathbb{R}$ is the estimate of the average. In a vector notation, (\ref{eq:consensus_smci1}) can be written as 
\begin{equation}
\begin{split}
\dot{\eta}(t) &= -\alpha\mathcal{B}^\mathrm{T}(t)\dot{z}(t)-\rho\tanh\{c \mathcal{B}^\mathrm{T}(t)\gamma(t)\}, \\
\gamma(t) &= \mathcal{B}(t)\eta(t)+z(t), \quad \eta(t_0) = \eta_{0},
\end{split}
\label{eq:consensus_smc1}
\end{equation}
where $\eta(t)=[{\eta_1}, \cdots, {\eta_N}]^\mathrm{T}$, $\gamma(t)=[{\gamma_1}, \cdots, {\gamma_N}]^\mathrm{T}$, $z(t)=[{z_1}, \cdots, {z_N}]^\mathrm{T}$, $\dot{z}(t)=[{\dot{z}_1}, \cdots, {\dot{z}_N}]^\mathrm{T}$, $\mathcal{B}(t)$ is the incidence matrix, and the $tanh(.)$ is defined component wise.

Let $\xi(t) = \eta(t)+\alpha\mathcal{B}^\mathrm{T}z(t)$. Then, the DAC algorithm (\ref{eq:consensus_smc1}) can be transformed to the following equivalent algorithm
\begin{equation}
\begin{split}
\dot{\xi}(t) &= -\rho\tanh\{c \mathcal{B}^\mathrm{T}(t)\gamma(t)\}, \quad \xi(t_0) = \xi_{0}, \\
\gamma(t) &= \mathcal{B}(t)\xi(t)+(I-\alpha \mathcal{B}(t)\mathcal{B}^\mathrm{T}(t))z(t),
\end{split}
\label{eq:consensus_smc1Tr}
\end{equation}
 which can be implemented without the knowledge of derivative information of reference signals. The agent-wise representation of the consensus algorithm (\ref{eq:consensus_smc1Tr}) is given as
\begin{equation}
\begin{split}
\dot{\xi}_{ij}^+(t) &= -\rho\tanh\{c(\gamma_{i}(t)-\gamma_{j}(t))\}\\
\dot{\xi}_{ij}^-(t) &=
-\rho\tanh\{c(\gamma_{j}(t)-\gamma_{i}(t))\}\\
\gamma_{i}(t) &= \sum_{j\in \mathcal{N}_i}{\xi}_{ij}^+(t)-\sum_{j\in \mathcal{N}_i}{\xi}_{ij}^-(t)+(1-\alpha d_i)z_{i}(t)\\
&\qquad \qquad \qquad \qquad \qquad \qquad-\alpha \sum_{j \in \mathcal{N}_i}z_j(t),\\ \xi_{ij}^+(t_0) &= \xi_{ij_{0}}, \xi_{ij}^-(t_0) = \xi_{ij_{0}}, c\geq 1, j\in \mathcal{N}_i,
\end{split}
\label{eq:consensus_smc1ti}
\end{equation}
where $d_i$ is the number of degrees of agent (node) $i$ in the network.
 Similar to the consensus algorithm (\ref{eq:consensus_smc}), by using $\textrm{tanh}$ function,  the consensus algorithm (\ref{eq:consensus_smc1}) does not suffer from chattering effect. Also, the consensus protocol (\ref{eq:consensus_smc1}) is guaranteed to asymptotically converge to the average $\bar{z}(t)$ with zero-steady-state error. This claim is formally stated in the following theorem. 
 \begin{theorem}
 For the balanced and strongly connected bidirectional graph $\mathcal{G}(t)$ and  time-varying reference signals $z_i(t)$, $i =\{1,\cdots,N\}$ with bounded first derivatives, the robust DAC algorithm in (\ref{eq:consensus_smc1}) guarantees that the consensus error $\tilde{\gamma}(t) = |\gamma(t) - \frac{1_N1_N^\mathrm{T}}{N}z(t)|$, asymptotically converges to zero for any $\eta_0$ if and only if $\rho > 0$, $c \geq 1$, and $\alpha I-(B(t)B^\mathrm{T}(t))^+ = 0$.
 \end{theorem}
 \begin{proof}
 From (\ref{eq:consensus_smc1}), the consensus error and the consensus error dynamics can be written as 
 \begin{equation}
 \begin{split}
 \tilde{\gamma}(t) &= \gamma(t) - \frac{1_N1_N^\mathrm{T}}{N}z(t) = \mathcal{B}(t)\eta(t)+Mz(t)\\
 \dot{\tilde{\gamma}}(t) &= \mathcal{B}(t)\dot{\eta}(t)+M\dot{z}(t).
 \end{split}
 \end{equation}
 Then, $\sum_{i=1}^N\gamma_i(t) = \mathbf{1}_N^\mathrm{T}\tilde{\gamma}(t)  = 0$ and $\mathbf{1}_N^\mathrm{T}\dot{\tilde{\gamma}}(t)  = 0$. Now, consider a candidate Lyapunov function
 \begin{equation}
 V = \frac{1}{2}\tilde{\gamma}^\mathrm{T}(t)\tilde{\gamma}(t).
 \end{equation}
 Taking a derivative of $V$, we have 
\begin{multline}
\dot{V} = \tilde{\gamma}^\mathrm{T}(t)\mathcal{B}(t)\dot{\eta}(t)\\
+\tilde{\gamma}^\mathrm{T}(t)\mathcal{B}(t)\mathcal{B}^\mathrm{T}(t)(\mathcal{B}(t)\mathcal{B}^\mathrm{T}(t))^+\dot{z}(t).
\end{multline}
Letting $\vartheta(t) = \mathcal{B}^\mathrm{T}(t)\tilde{\gamma}(t)=\mathcal{B}^\mathrm{T}(t)\gamma(t)$ and using $\dot{\eta}$ from (\ref{eq:consensus_smc1}), we have
\begin{multline}
\dot{V} = -\rho\vartheta^\mathrm{T}(t)\tanh(c\vartheta(t))\\
-\vartheta^\mathrm{T}(t)\mathcal{B}^\mathrm{T}(t)(\alpha I - (\mathcal{B}(t)\mathcal{B}^\mathrm{T}(t))^+)\dot{z}(t)
\end{multline}
Now, $\forall \tilde{\gamma}(t) \in \mathbb{R}^N$, $\tilde{\gamma}(t) \neq \textbf{0}$, $\dot{V} < 0$ if $\rho$ and $\alpha$ are selected as
\begin{equation}
\rho > 0, \quad \alpha I-(B(t)B^\mathrm{T}(t))^+ = 0.
\label{eq:cond}
\end{equation}
Therefore, employing LaSalle's theorem, the consensus protocol (\ref{eq:consensus_smc1}) guarantees the asymptotic convergence of the estimate of the average $\gamma(t)$ to $\bar{z}(t)$. Note that we can also conclude the condition (\ref{eq:cond}) using the consensus protocol (\ref{eq:consensus_smc1Tr}). 
 \end{proof}

\section{Simulation results}
\label{sec:form_and_tra_simu}
In this section, we provide the simulation results to verify the performance of the proposed consensus algorithms.
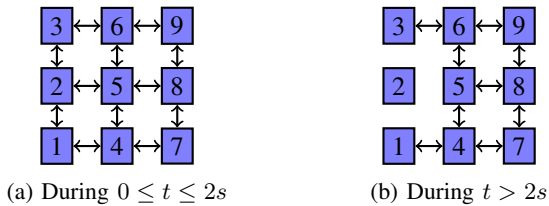
\begin{figure}
\centering
\begin{subfigure}[b]{0.23\textwidth}
\centering
\tikzstyle{dot}=[rectangle, draw, fill=blue!50,
                        minimum height=10pt, minimum width=10pt]
\begin{tikzpicture}[thick,scale=0.4]
        \node[dot](1) at (0,0) {1};
        \node[dot](4)at (2,0) {4};
        \node[dot](7) at (4,0) {7};
        \node[dot](2) at (0,2) {2};
        \node[dot](3) at (0,4) {3};
        \node[dot](5) at (2,2) {5};
        \node[dot](6) at (2,4) {6};
        \node[dot](8) at (4,2) {8};
        \node[dot](9) at (4,4) {9};
        \draw [<->](1.east)  -- (4.west);
        \draw [<->](1.north)  -- (2.south);
        \draw [<->](2.north)  -- (3.south);
        \draw [<->](3.east)  -- (6.west);
        \draw [<->](4.north)  -- (5.south);
        \draw [<->](7.north)  -- (8.south);
        \draw [<->](5.north)  -- (6.south);
        \draw [<->](6.east)  -- (9.west);
        \draw [<->](9.south)  -- (8.north);
        \draw [<->](2.east)  -- (5.west);
        \draw [<->](4.east)  -- (7.west);
        \draw [<->](8.west)  -- (5.east); 
\end{tikzpicture}
\caption{During $0 \leq t \leq 2s$}
\end{subfigure}
\hfill
\begin{subfigure}[b]{0.23\textwidth}
\centering
\tikzstyle{dot}=[rectangle, draw, fill=blue!50,
                        minimum height=10pt, minimum width=10pt]
\begin{tikzpicture}[thick,scale=0.4]
        \node[dot](1) at (0,0) {1};
        \node[dot](4)at (2,0) {4};
        \node[dot](7) at (4,0) {7};
        \node[dot](2) at (0,2) {2};
        \node[dot](3) at (0,4) {3};
        \node[dot](5) at (2,2) {5};
        \node[dot](6) at (2,4) {6};
        \node[dot](8) at (4,2) {8};
        \node[dot](9) at (4,4) {9};
        \draw [<->](1.east)  -- (4.west);
        \draw [<->](3.east)  -- (6.west);
        \draw [<->](4.north)  -- (5.south);
        \draw [<->](7.north)  -- (8.south);
        \draw [<->](5.north)  -- (6.south);
        \draw [<->](6.east)  -- (9.west);
        \draw [<->](9.south)  -- (8.north);
        \draw [<->](4.east)  -- (7.west);
        \draw [<->](8.west)  -- (5.east); 
\end{tikzpicture}
\caption{During $t > 2s$}
\end{subfigure}
\caption{ The underlying communication graph representing the network of agents-- initially, all the agents are on the same network during  $0 \leq t \leq 2s$. After two seconds, the \emph{agent 2} communication module fails, and $agent 2$ will not be able to communicate with all its neighbors creating two sub-networks.  }
\label{fig:example_consensus}
\vspace{-0.625cm}
\end{figure}
Consider a network of nine agents with the underlying communication graph given in Figure \ref{fig:example_consensus}. Let the agents' time-varying reference signals, $z_i(t)$, $i = \{1,\cdots,9\}$ be given as
\begin{equation}
\begin{split}
z_1 &= 5\cos(t), z_2 = 4\cos(t), z_3 = 3\cos(t),\\
z_4 &= 2\cos(t), z_5 = \cos(t), z_6 = -\cos(0.01t),\\
z_7 &= -2\cos(0.01t), z_8 = -3\cos(0.01t), \\
z_9 &= -4\cos(0.01t).
\end{split}
\end{equation}
We let the \emph{agent 2} to fail to communicate with its neighbours for all times after $2s$ as shown in Figure \ref{fig:example_consensus} to verify the robustness of the proposed algorithm to agents leaving the network. This action will create two sub-networks. The first sub-network has only a node and no edges, while the other sub-network has the remaining 8-nodes and the edges among them. We present the estimate of the average of the agents' time-varying signals implementing the consensus algorithms (\ref{eq:consensus_smci}) and (\ref{eq:consensus_smc1ti}) in Figures \ref{fig:firstc}-\ref{fig:seconde}, respectively. The consensus algorithm (\ref{eq:consensus_smci}) is proved to guarantee the boundedness of the estimation error. Choosing the parameters of the consensus algorithm  (\ref{eq:consensus_smci}) as $c = 1$, and $\rho = 16$ with sampling time of $0.01s$, the distributed estimate of the average of the agents' time varying reference signals and the consensus errors are presented in Figures \ref{fig:firstc} and \ref{fig:firste}, respectively. In contrast, the consensus algorithm (\ref{eq:consensus_smc1ti}) guarantees the asymptotic convergence of the consensus error.
Choosing the parameters of the consensus algorithm  (\ref{eq:consensus_smc1ti}) as $c = 4$, $\alpha = 0.16$, and $\rho = 4.1$ with sampling time of $0.01s$, the distributed estimate of the average of the agents' time varying reference signals and the consensus estimation errors are presented in Figures \ref{fig:secondc} and \ref{fig:seconde}, respectively.

 We compare our results to the robust discontinuous consensus protocol in \cite{george2017robust} which is proved to converge in finite time. However, the precise implementation of the discontinuous switching estimator input signal requires very tiny sampling time and it is impossible to guarantee zero steady-state error in practice. Now, choosing the parameters of the consensus algorithm in \cite{george2017robust} as $\alpha_i = 10$ with sampling time of of $0.0001s$, the estimate of the average and the consensus estimation errors are presented in Figures \ref{fig:george1c} and \ref{fig:george1e}, respectively. The results show that unlike the consensus algorithms proposed in this paper, the consensus algorithms in \cite{george2017robust} requires very low sampling times. For comparably higher sampling times, the results of the consensus algorithm in \cite{george2017robust} deteriorates as demonstrated in Figures \ref{fig:george2c}-\ref{fig:george2e}. The simulation results in Figures \ref{fig:george2c}-\ref{fig:george2e} are generated by choosing the parameters of the consensus algorithms as  $\alpha_1 = 5.7$, $\alpha_2 = 4.6$, $\alpha_3 = 3.4$, $\alpha_4 = 2.3$, $\alpha_5 = 1.2$, $\alpha_6 = 1.2$, $\alpha_7 = 2.3$, $\alpha_8 = 3.4$, and $\alpha_9 = 4.6$ with the sampling time of $0.01s$. In fact, the authors in \cite{george2017robust} observed these issues in their algorithm and proposed a boundary layer approximation to circumvent the problem. Nevertheless, with the boundary layer approximation, the convergence proof will only guarantee a bounded steady-state error in a finite time. However, designed by similar approaches to \cite{george2017robust}, our second algorithm guarantees asymptotic convergence of the average consensus error while avoiding any chattering effects.
\begin{figure*}[!h]
\begin{subfigure}[b]{1.0\columnwidth}
 \centering
\includegraphics[scale=0.225]{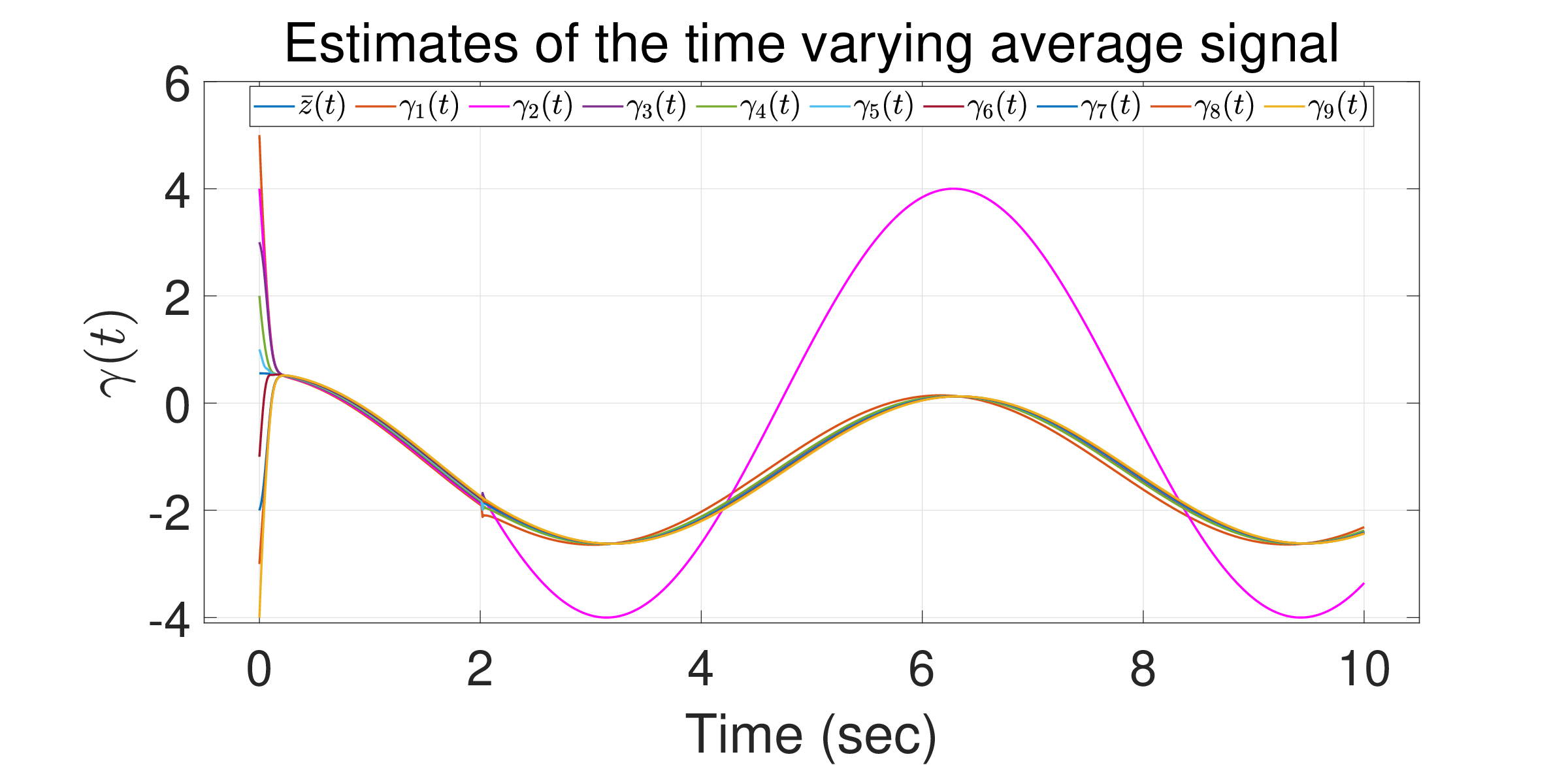}
\caption{}
\label{fig:firstc}
\end{subfigure}
\begin{subfigure}[b]{1.0\columnwidth}
\centering
\includegraphics[scale=0.225]{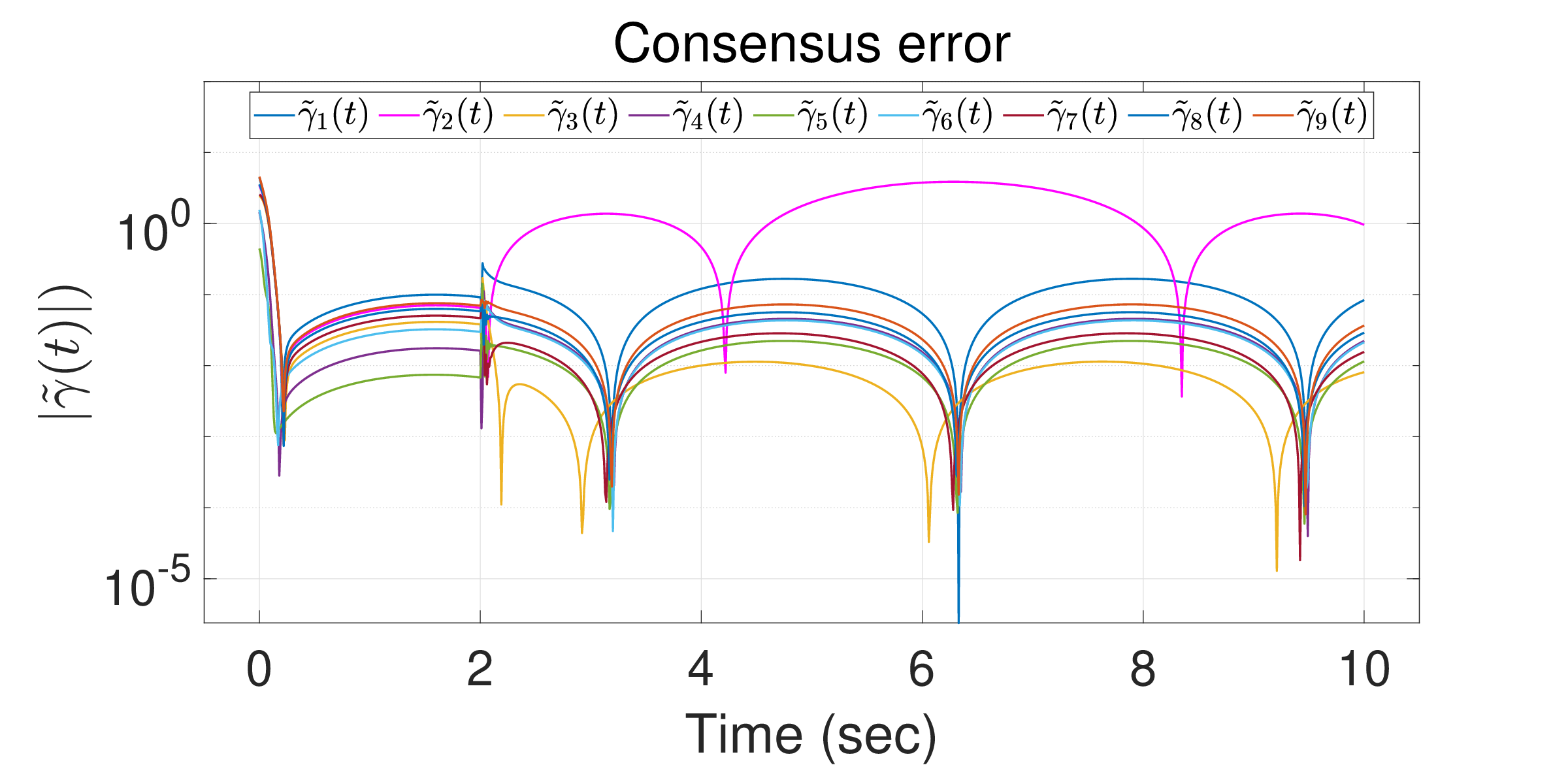}
\caption{}
\label{fig:firste}
\end{subfigure}\\
\begin{subfigure}[b]{1.0\columnwidth}
\centering
\includegraphics[scale=0.225]{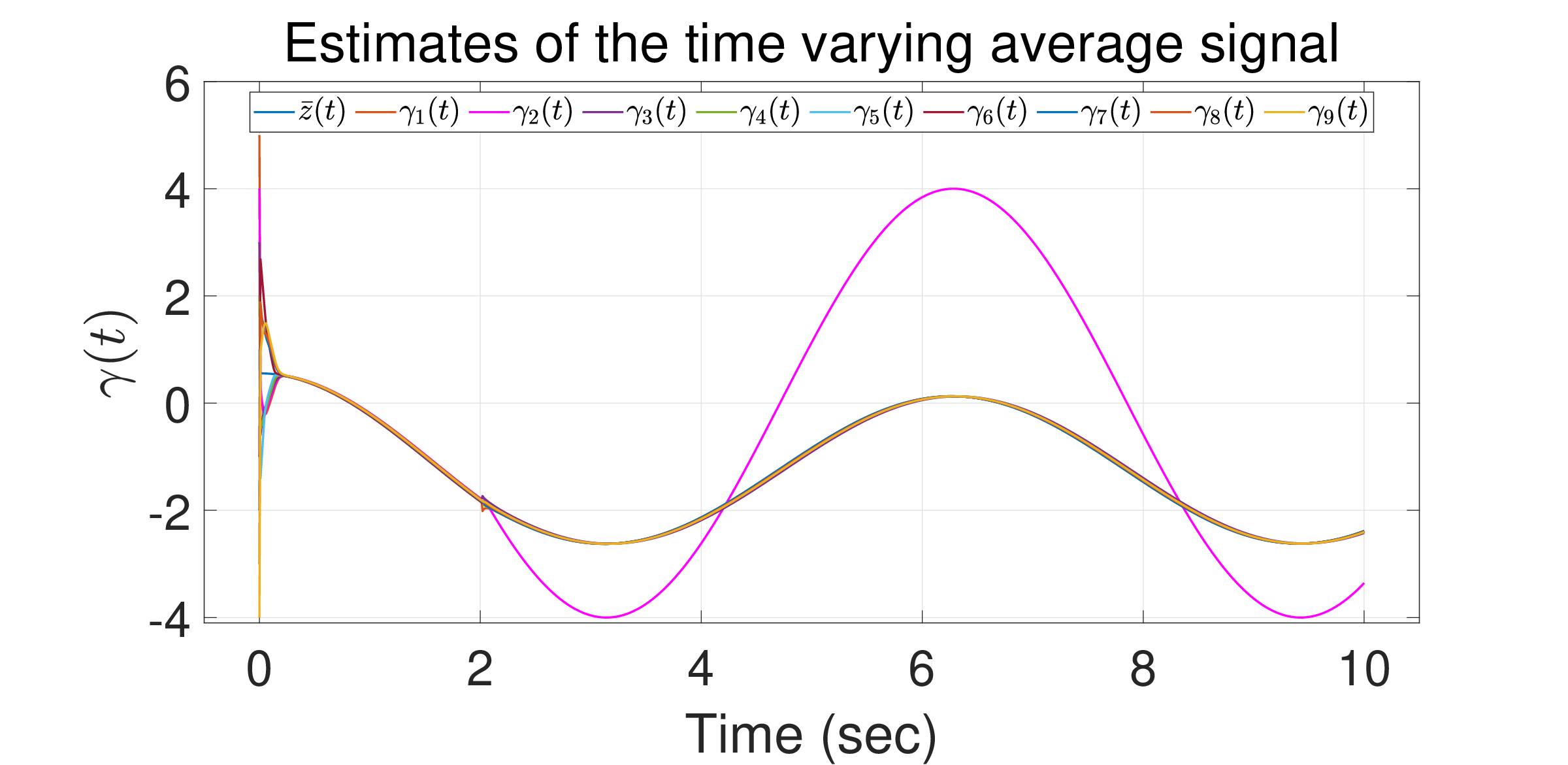}
\caption{}
\label{fig:secondc}
\end{subfigure}
\begin{subfigure}[b]{1.0\columnwidth}
\centering
\includegraphics[scale=0.225]{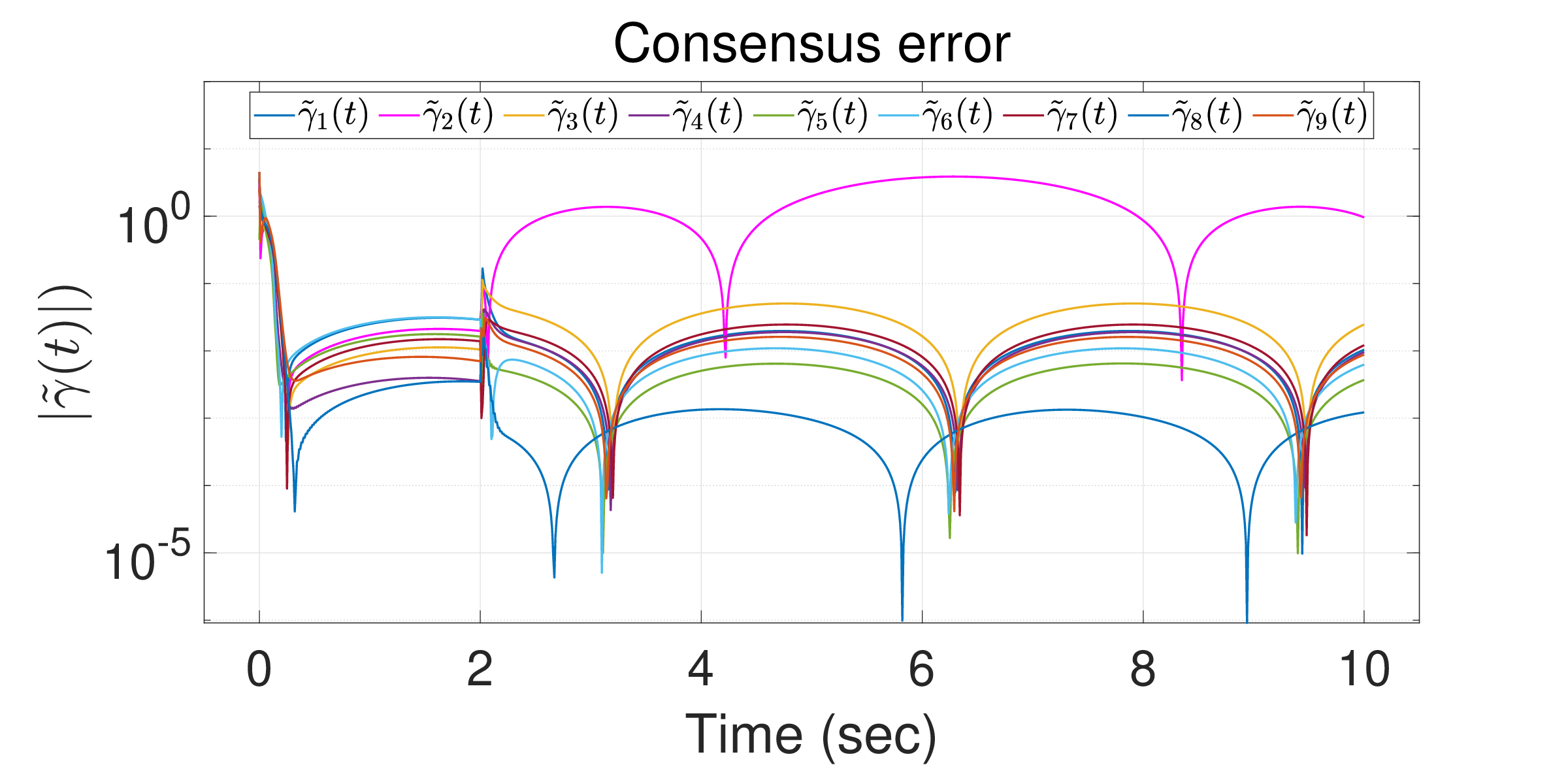}
\caption{}
\label{fig:seconde}
\end{subfigure}\\
\begin{subfigure}[b]{1.0\columnwidth}
\centering
\includegraphics[scale=0.225]{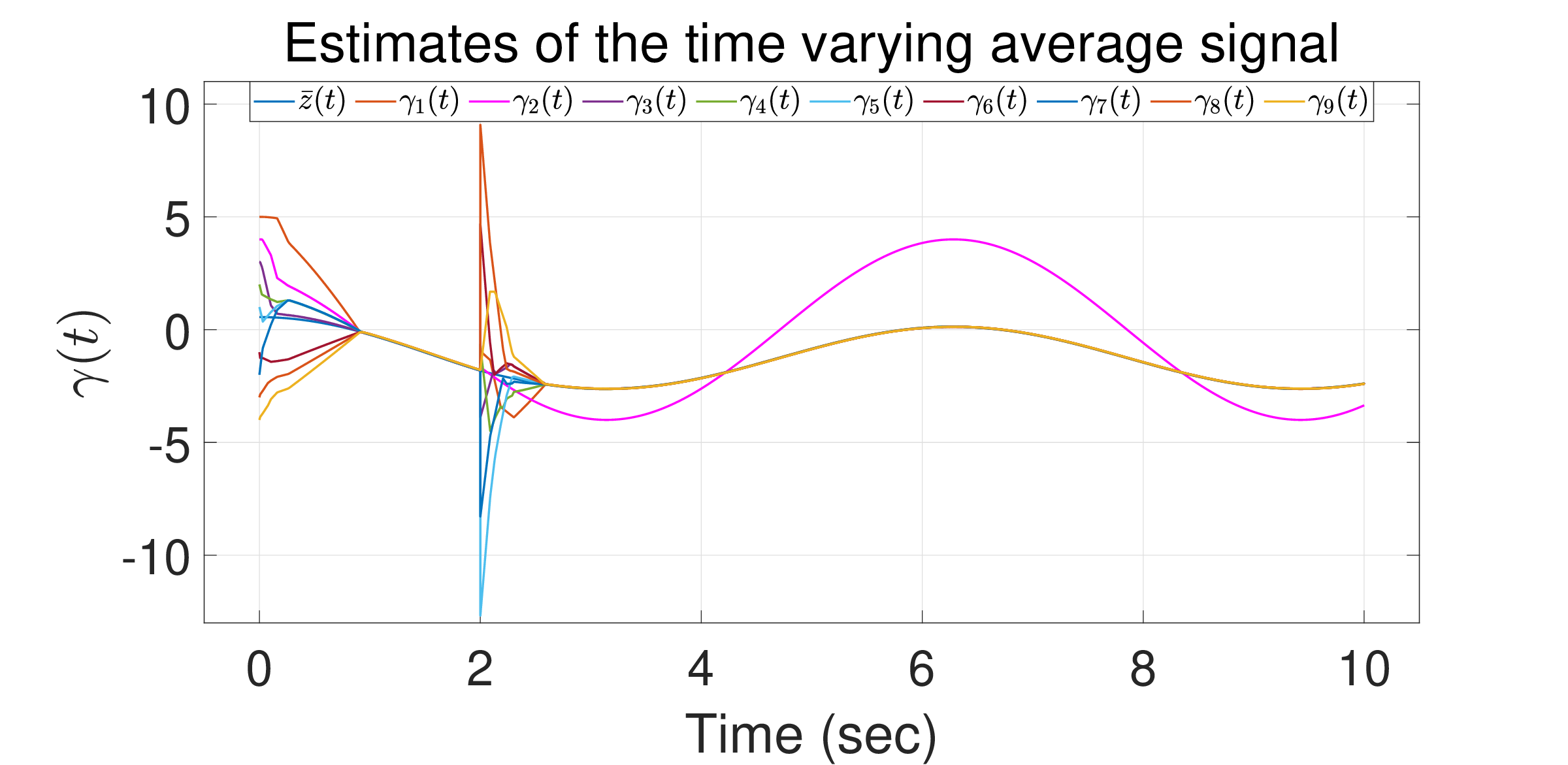}
\caption{}
\label{fig:george1c}
\end{subfigure}
\begin{subfigure}[b]{1.0\columnwidth}
\centering
\includegraphics[scale=0.225]{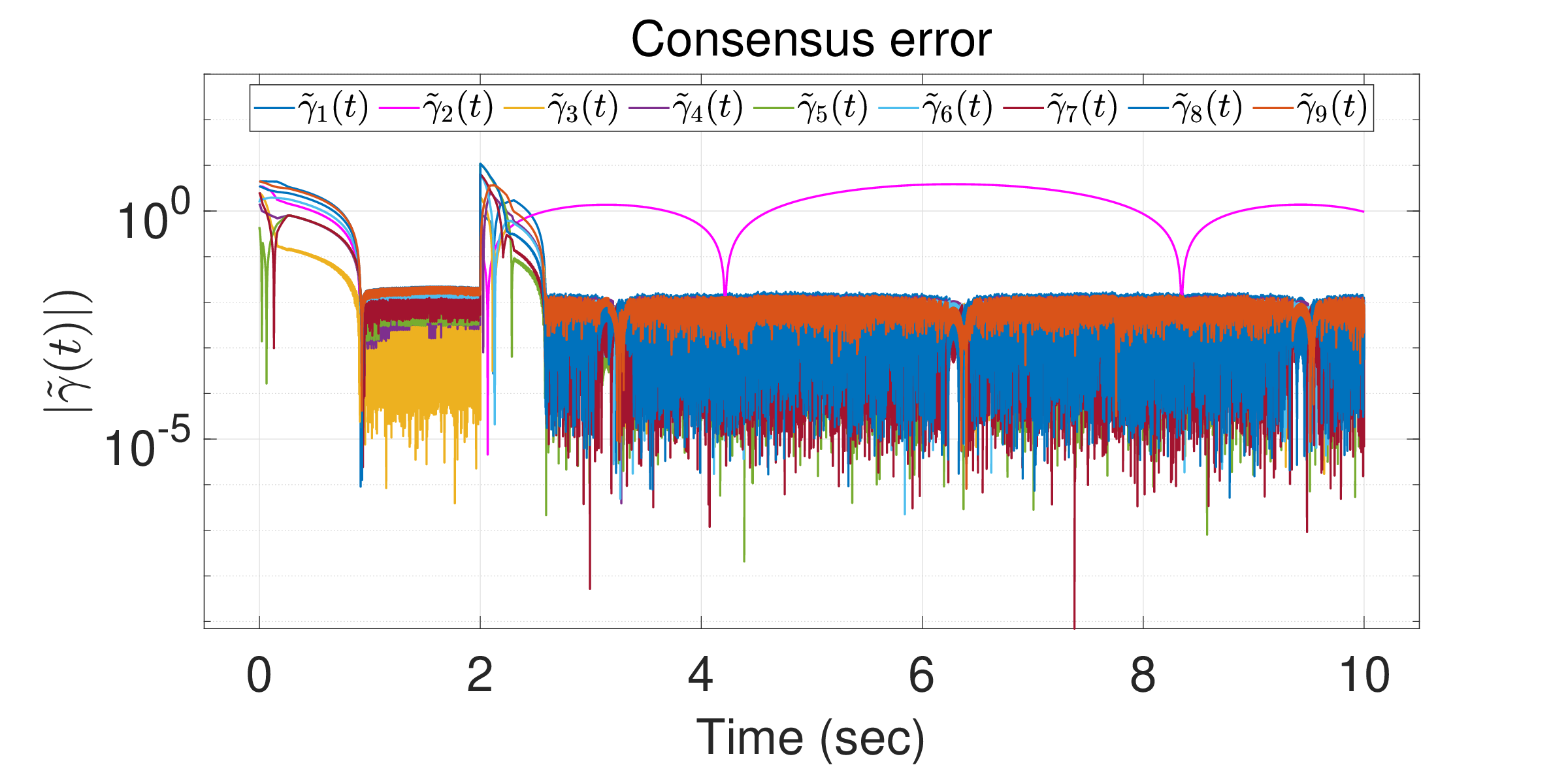}
\caption{}
\label{fig:george1e}
\end{subfigure}\\
\begin{subfigure}[b]{1.0\columnwidth}
\centering
\includegraphics[scale=0.225]{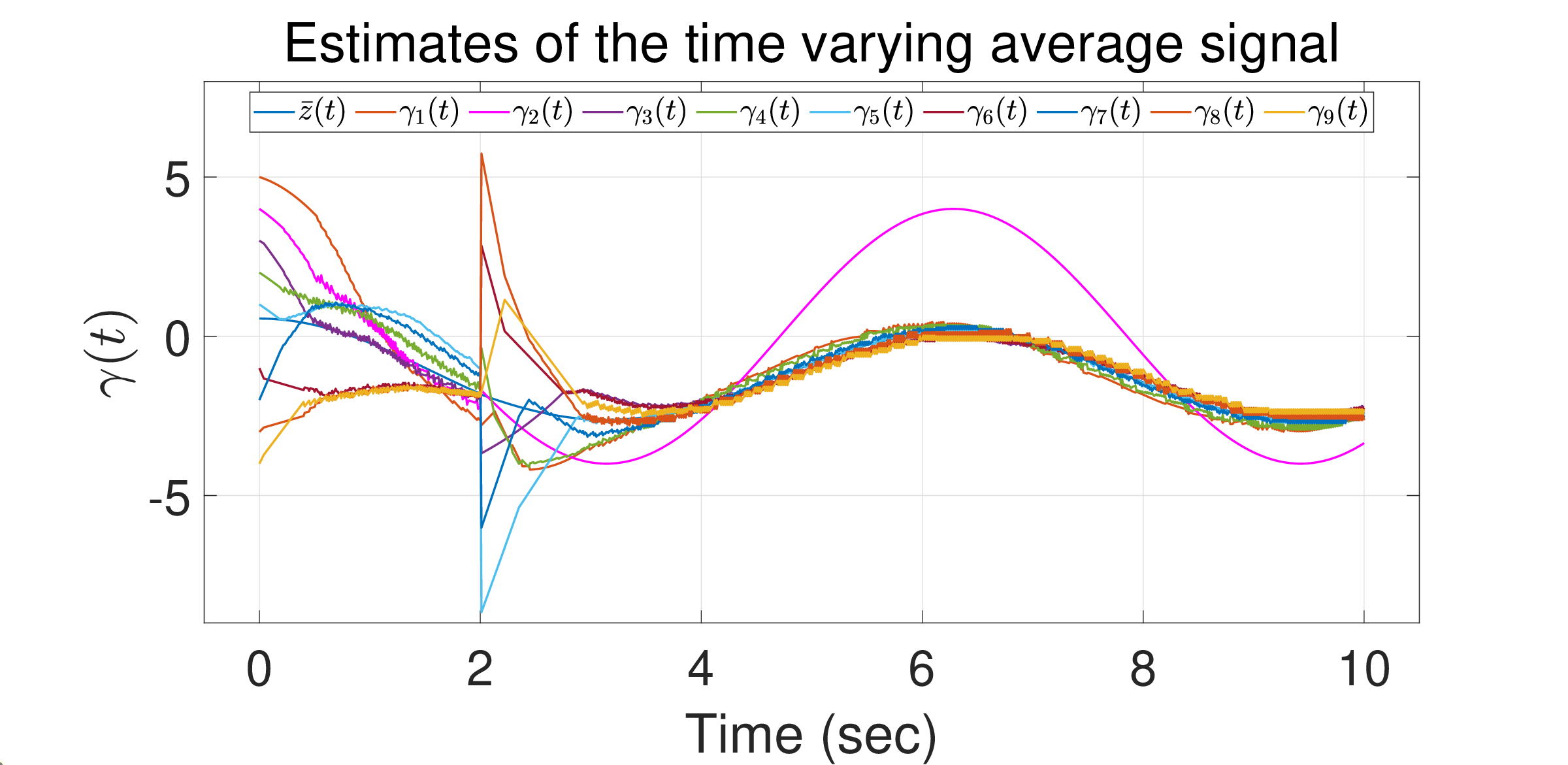}
\caption{}
\label{fig:george2c}
\end{subfigure}
\begin{subfigure}[b]{1.0\columnwidth}
\centering
\includegraphics[scale=0.225]{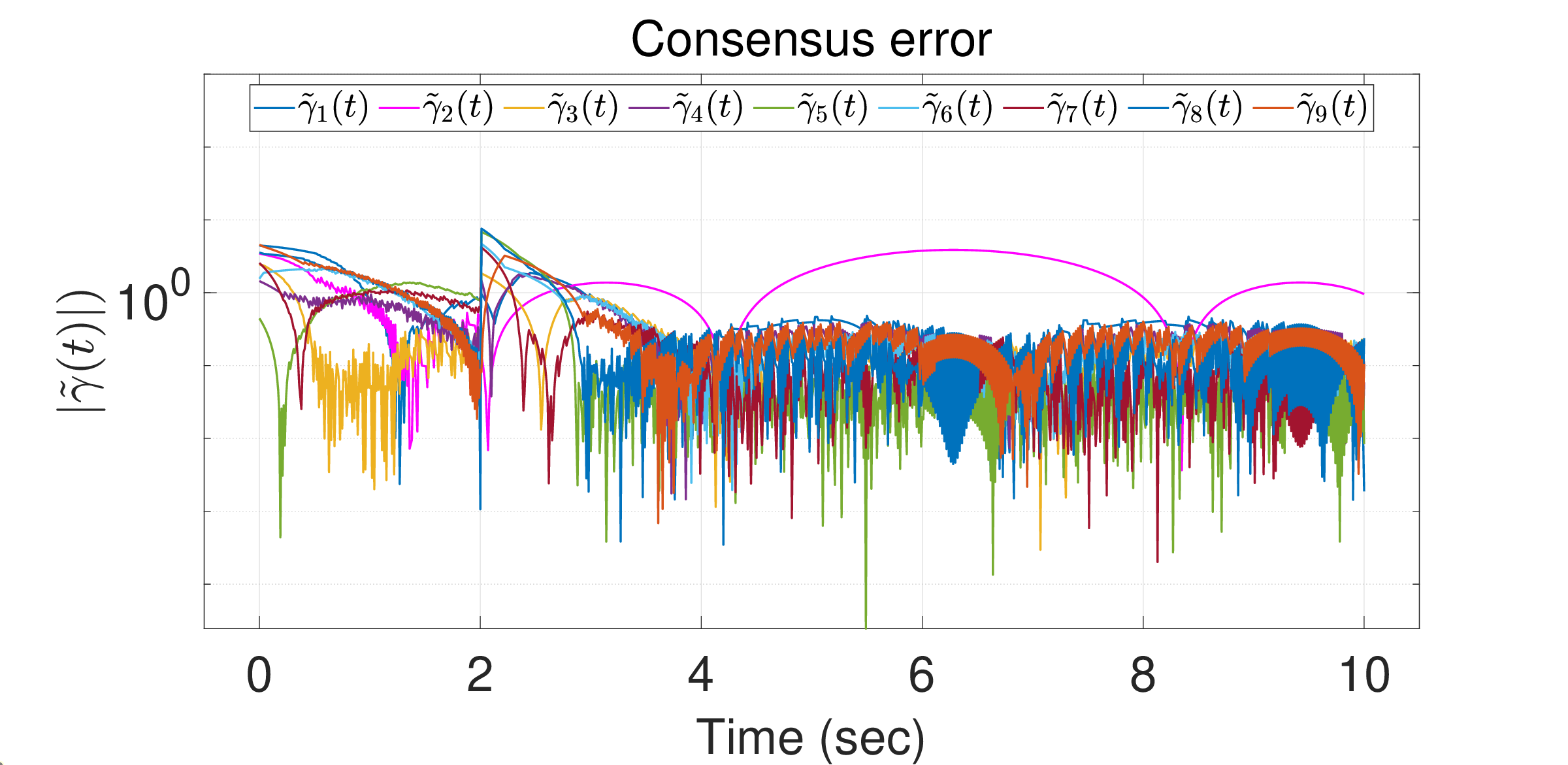}
\caption{}
\label{fig:george2e}
\end{subfigure}
\caption{Performance comparisons of the proposed DAC algorithms with the DAC algorithm proposed in \cite{george2017robust}: (a) Estimate of average of multiple time-varying signals using our  proposed consensus algorithm (\ref{eq:consensus_smci}), (b) Log plot of estimation error of our proposed consensus algorithm (\ref{eq:consensus_smci}), (c) Estimate of average of multiple time-varying signals using our proposed consensus algorithm (\ref{eq:consensus_smc1ti}), (d) Log plot of estimation error of our proposed consensus algorithm (\ref{eq:consensus_smc1ti}), (e) Estimate of average of multiple time-varying signals using the consensus algorithm in \cite{george2017robust} with sampling time of $0.0001s$, (f) Log plot of estimation error of the consensus algorithm in \cite{george2017robust} with sampling time of $0.0001s$, (g) Implementation of the consensus algorithm in \cite{george2017robust} with sampling time of $0.01 s$, and (h) Log plot of estimation error of the consensus algorithm in \cite{george2017robust} implemented with sampling time of $0.01s$.}
\end{figure*}

\section{conclusion}
\label{sec:form_and_tra_conc}
In this paper, we proposed two DAC algorithms that allow a network of agents to estimate the average of their time-varying reference signals cooperatively. The algorithms are robust to agents joining and leaving the network, at the same time, remove the chattering phenomena that arise in many non-linear consensus protocols. Further, we provided the convergence and robustness analysis of the proposed consensus protocols utilizing Lyapunov functions. The convergence analysis shows that the first algorithm guarantees bounded steady-state error, while the second algorithm guarantees asymptotic convergence to zero steady-state error. We also provided a discrete-time implementation and demonstrated a simulation example to show the effectiveness of the proposed consensus protocols. Future work focuses on extending the algorithm to directed graph topology in the presence of delays, and sensor and model uncertainties.
\section*{Acknowledgment}
This research is supported by Air Force Research Laboratory and OSD under agreement number FA8750-15-2-0116 as well as the National Science Foundation under award number 1832110.
The U.S. Government is authorized to reproduce and distribute reprints for Governmental purposes notwithstanding any copyright notation thereon. The views and conclusions contained herein are those of the authors and should not be interpreted as necessarily representing the official policies or endorsements, either expressed or implied, of Air Force Research Laboratory,  OSD, National Science Foundation, or the U.S. Government.

\bibliographystyle{IEEEtran}
\bibliography{Conference_ver}
\end{document}